\documentclass[11pt,a4paper]{article}

\usepackage[dvips]{graphicx}
\graphicspath{{eps/}}

\usepackage[cp1251]{inputenc}
\usepackage[T2A]{fontenc}
\usepackage[english]{babel}

\usepackage{animate}

\usepackage{amssymb}
\usepackage{amsxtra}
\usepackage{amsthm}
\usepackage{color}
\usepackage{hyperref}
\usepackage[in]{fullpage}
\numberwithin{equation}{section}
\theoremstyle{plain}
\newtheorem{theorem}{Теорема}

\theoremstyle{remark}

\def\sgrad{\mathrm{sgrad}\,}
\def\trace{\mathrm{trace}\,}
\newcommand{\rank}{\mathop{\rm rank}\nolimits}

\begin{document}

\title{Phase Topology of Two Vortices of the Identical Intensities in Bose-Einstein Condensate}

\author{P.\,E.~Ryabov, S.\,V.~Sokolov}

\date{}

\maketitle

\begin{abstract}
Completely Liouville integrable Hamiltonian system with two deg\-rees of freedom, describing the dynamics of two vortex filaments in a Bose- Einstein condensate enclosed in a cylindrical trap, is considered. For the system of two vortices with identical intensities detected bi\-fur\-cation of three Liouville tori into one. Such a bifurcation was found in the integrable case of Goryachev-Chaplygin-Sretensky in the rigid body dynamics.

\vspace{3mm}

Key words: vortex dynamics, Bose-Einstein condensate, completely integrable Hamiltonian systems, bifurcation diagram of momentum mapping, bifurcations of Liouville tori

\end{abstract}
{

\parindent=0mm

UDC 532.5.031, 517.938.5

MSC 2010: 76M23, 37J35,  37J05, 34A05

------------------------------------------------------------

Received on 31 December 2018.

------------------------------------------------------------

The work of P.\,E.\,Ryabov was supported by RFBR grants 16-01-00170
and 17-01-00846. The work of S.\,V.\,Sokolov was carried out at MIPT under project 5–100 for state
support for leading universities of the Russian Fede\-ration and also partially support by RFBR grants
16-01-00809, 16-01-00170 and 18-01-00335.

------------------------------------------------------------

Ryabov Pavel Evgen`evich

PERyabov@fa.ru

Financial University under the Government of the Russian Federation\\
Leningradsky prosp. 49, Moscow, 125993 Russia

Institute of Machines Science, Russian Academy of Sciences\\
Maly Kharitonyevsky per. 4, Moscow, 101990 Russia

Udmurt State University\\
ul. Universitetskaya 1, Izhevsk, 426034 Russia

---------------------------------------------------------------

Sokolov Sergei Victorovich

sokolovsv72@mail.ru

Moscow Institute of Physics and Technology (State University)\\
9 Institutskiy per., Dolgoprudny, Moscow Region, 141701, Russian Federation

Institute of Machines Science, Russian Academy of Sciences\\
Maly Kharitonyevsky per. 4, Moscow, 101990 Russia

}

\newpage

\tableofcontents

\section{Introduction}
\setcounter{equation}{0}

The mainstream of the vortex analytical dynamics is an integrable models of point vortices on a plane. Studies of the dynamics of vortices in a quantum liquids, have shown that quantum vortices behave approximately the same as thin vortex filaments in a classical perfect fluids. A special place is occupied by the vortex structures in the Bose\,-- Einstein condensate obtained for ultracold atomic gases \cite{fett2009}. This article will be concerned with a mathematical model of the dynamics of two vortex filaments in a Bose\,-- Einstein condensate enclosed in a harmonic trap  \cite{kevrekPhysLett2011}, \cite{kevrek2013}, \cite{kevrikidis2014}. This model leads to a completely Liouville integrable Hamiltonian system with two degrees of freedom, and for this reason, topological methods used in such systems can be applied. Topological methods were successfully used for investigation of the stability problem of absolute and relative choreographies \cite{borkil2000}, \cite{BorMamSokolovskii2003}, \cite{bormamkil2004}, \cite{kilinbormam2013}, \cite{BorSokRyab2016}. These motions in integrable models, as a rule, correspond to the values of the constant first integrals, for which the integrals, considered as functions of phase variables, turn out to be dependent in the sense of the linear dependence of the differentials. The main role in the study of such dependence is played by the bifurcation diagram of the momentum map.

This publication is devoted to the integrable perturbation of the model considered in \cite{SokRyabRCD2017}, \cite{sokryab2018}. In this paper, the bifurcation diagram is explicitly determined and bifurcations of Liouville tori are investigated. In the case of a vortex pair consisting of vortices of the identical intensity, a bifurcation of three tori into one is detected for some values of physical parameters. This bifurcation was previously encountered in the works of M.\,P.~Kharlamov in studying the phase topology of the integrable Chaplygin-Goryachev-Sretensky case in the dynamics of a rigid body \cite{Kharlamov1988} and as one of the features in the form of a 2-atom of a singular layer of Liouville foliation in the works of A.\,T.~Fomenko, A.\,V.~Bolsinov, S.\,V.~Matveev \cite{bolsmatvfom1990}. In the work of A.\,A.~Oshem\-kov and M.\,A.~Tuzhilin \cite{oshtuzh2018}, devoted to the splitting of saddle singularities, such a bifurcation turned out to be unstable and its perturbed foliations, one of which is realized in the integrable model under consideration, are given.

\section{Model and Definitions}
Here we following the original works \cite{kevrekPhysLett2011}, \cite{kevrek2013},
\cite{kevrikidis2014} in the description of the model.
Let us consider $N$ interacting vortices in a Bose-Einstein condensate enclosed in a harmonic trap and let $(x_k,y_k)$ is the position of the $k$-th vortex, $r_k=\sqrt{x_k^2+y_k^2}$.  A single vortex $(x_k,y_k)$ in a harmonic trap is well known to precess around the center of the trap with  the frequency $\omega_{\rm{pr}}$ which can be approximated by
$\omega_{\textrm{pr}}=\omega_{\textrm{pr}}^{0}/(1-r_k^2/R_{\textrm TF}^2)$, where the frequency at the
trap center is $\omega_{\textrm{pr}}^{0}=\ln\bigl(A\frac{\mu}{\Omega}\bigr)/R_{\textrm{TF}}^2$, $\mu$ is the chemical potential, $R_{\textrm{TF}}=\sqrt{2\mu}/\Omega$ is the so-called Thomas-Fermi~(TF) radius, $A=2\sqrt{2}\pi$ is a numerical constant, and $\Omega=\omega_r/\omega_z$, here $\omega_r$ and $\omega_z$ are the confining radial and axial frequencies of the harmonic trap, respectively. On the other hand, in the absence of a harmonic trap, two interacting vortices will rotate around each other with a frequency of $\omega_{\textrm{vort}} = B/r_{kj}^2$, where $r_{kj}=\sqrt{(x_k-x_j)^2+(y_k-y_j)^2}$ is the distance between the vortices and $B$ is a constant factor. If $(x_k,y_k)$ is the position of the $k$-th vortex, the corresponding the mathematical model of the dynamics of $N$ interacting vortices in a Bose-Einstein condensate (BEC) enclosed in a harmonic trap is described by following the system of differential equations \cite{kevrekPhysLett2011}, \cite{kevrek2013},
\cite{kevrikidis2014}:
\begin{equation}
\label{d0}
\begin{array}{l}
\displaystyle{\dot x_k=-\Gamma_k\omega_{\textrm{pr}}y_k-\frac{B}{2}\sum\limits_{j\neq k}^{N}\,\Gamma_j\frac{y_k-y_j}{r_{kj}^2},}\\[3mm]
\displaystyle{\dot y_k=\Gamma_k\omega_{\textrm{pr}}x_k+\frac{B}{2}\sum\limits_{j\neq k}^{N}\,\Gamma_j\frac{x_k-x_j}{r_{kj}^2},}
\end{array}
\end{equation}
where $\Gamma_k$ is the charge of the $k$-th vortex, ($k=1,\ldots,N$) and $N$ is the total
number of interacting vortices.

For convenience, following to the paper \cite{kevrikidis2014} we can further rescale time to the period of the single vortex
precessing near the center of the trap ($\tau=t\omega_{\textrm pr}^0$) and pass to dimensionless variables
using the relations
\begin{equation}
\label{d1}
x_k=\tilde{x}_kR_{\textrm TF},\quad y_k=\tilde{y}_kR_{\textrm TF}.
\end{equation}

After rescaling \eqref{d1}, the equations of motion \eqref{d0} are written as
\begin{equation}
\label{d2}
\begin{array}{l}
\displaystyle{x_k^\prime=-\Gamma_k\frac{y_k}{1-r_k^2}-c\sum\limits_{j\neq k}^{N}\,\Gamma_j\frac{y_k-y_j}{r_{kj}^2},}\\[3mm]
\displaystyle{y_k^\prime=\Gamma_k\frac{x_k}{1-r_k^2}+c\sum\limits_{j\neq k}^{N}\,\Gamma_j\frac{x_k-x_j}{r_{kj}^2},}
\end{array}
\end{equation}
where the dimensionless parameter $c$ is defined by the formula
\begin{equation*}
c=\frac{B}{2\ln\bigl(A\frac{\mu}{\Omega}\bigr)}
\end{equation*}
and the prime $()^\prime$ in \eqref{d2} stands for $\tfrac{d}{d\tau}$.

The equations of motion \eqref{d2} can be represented in Hamiltonian form
\begin{equation}
\label{x1}
\Gamma_k x_k^\prime=\frac{\partial H}{\partial y_k}, \quad \Gamma_k y_k^\prime=-\frac{\partial H}{\partial x_k},\quad k=1,\ldots,N.
\end{equation}
with Hamiltonian
\begin{equation}
\label{x0}
\displaystyle{H=\frac{1}{2}\sum\limits_{k=1}^N\,\Gamma_k^2\ln(1-r_k^2)-\frac{c}{2}\sum\limits_{k=1}^N\sum\limits_{j<k}^N\,
\Gamma_k\Gamma_j\ln(r_{kj}^2).}
\end{equation}

In this paper follows hereafter, we restrict our study to the case  with two vortices ($N=2$)  of the identical intensities ($\Gamma_1=\Gamma_2=1$). In this case, the Hamiltonian \eqref{x0} will be written as
\begin{equation}
\label{x2}
\displaystyle{H=\frac{1}{2}\ln[1-(x_1^2+y_1^2)]+\frac{1}{2}\ln[1-(x_2^2+y_2^2)]-\frac{c}{2}\ln[(x_2-x_1)^2+(y_2-y_1)^2].}
\end{equation}

In the works \cite{SokRyabRCD2017}, \cite{sokryab2018} parameter $c$ was taken to be equal to one, however in a series of physical works \cite{kevrekPhysLett2011}, \cite{kevrek2013} for the parameter $c$ in the case of a vortex pair consisted of vortices with identical intensities ($\Gamma_1=\Gamma_2=1$) other values of $c=1.35$; $c=0.1$ were taken on the basis of experimental data. In this connection, it is of interest to study the phase topology, when the parameter $c$ takes any positive values.

The phase space $\cal P$ is given as a direct product of two open circles of radius $1$, with a set of vortex collisions punctured out
\begin{equation*}
\label{x3}
{\cal P}=\{(x_1,y_1,x_2,y_2)\,:\, x_1^2+y_1^2<1, x_2^2+y_2^2<1\}\smallsetminus \{x_1=x_2, y_1=y_2\}.
\end{equation*}
The Poisson structure on the phase space $\cal P$ is given in the standard form
\begin{equation}
\label{x4}
\{x_i,y_j\}=\frac{1}{\Gamma_i}\delta_{ij},
\end{equation}
where $\delta_{ij}$ is the Kronecker symbol.

System \eqref{x1} admits one additional first integral of motion -- \textit{moment of vorticity} ($N=2$, $\Gamma_1=\Gamma_2=1$)
\begin{equation}
\label{x5}
F=\Gamma_1(x_1^2+y_1^2)+\Gamma_2(x_2^2+y_2^2).
\end{equation}

The function $F$ with the Hamiltonian $H$ form on $\cal P$ a complete involutive set of integrals of the system \eqref{x1} ($N=2,\Gamma_1=\Gamma_2=1$). According to the Liouville-Arnold theorem, the regular surface of the common level of first integrals is a disconnected union of two-dimensional tori filled with conditionally periodic trajectories. We define \textit{integral mapping} ${\cal F}\,:\, {\cal P}\to {\mathbb R}^2$, setting $(f, h)={\cal F} (\boldsymbol\zeta)=(F(\boldsymbol\zeta), H(\boldsymbol\zeta))$. The mapping $\cal F$ is also called \textit{momentum map}. Denote by $\cal C$ the set of all critical points of the momentum maps, that is, points at which $\rank d{\cal F}(x) < 2$. The set of critical values $\Sigma = {\cal F}({\cal C}\cap {\cal P})$ is called \textit{bifurcation diagram}.

\section{Bifurcation Diagram}
To find the bifurcation diagram, we use the method of critical subsystems developed by M.\,P.~Kharlamov in integrable problems of the rigid body dynamics \cite{KharlamovRCD2014}. As an application, an analysis of the stability of critical trajectories (that is, nondegenerate singularities of rank 1 of the momentum map) is given by determining the type of trajectory (elliptic / hyperbolic) for each curve from the bifurcation set.

In later we confine to identical intensities ($\Gamma_1=\Gamma_2=1$).
Let ${\cal N}_1$ denote the closure of the set of solutions of the system
\begin{equation}
\label{y0}
x_1+x_2=0;\quad y_1+y_2=0
\end{equation}
and let ${\cal N}_2$ denote the closure of the set of solutions of the system
\begin{equation}
\label{y2}
F_1 = 0,\quad F_2 = 0,
\end{equation}
where
\begin{equation*}
\label{y1}
\begin{array}{l}
F_1=x_1y_2-y_1x_2,\\[3mm]
F_2=(x_1^2+x_2^2)(x_2^2+y_2^2)[x_1(x_2^2+y_2^2)-cx_2]+x_2[(c-2)(x_2^2+y_2^2)^2x_1^2+cx_2^2].
\end{array}
\end{equation*}

Then the theorem is valid.

\begin{theorem}
\label{t1}
A set $\cal C$ of critical points of the momentum map $\cal F$ is exhausted by a set of solutions
of the collection of systems \eqref{y0} and \eqref{y2}. The sets ${\cal N}_1$ and ${\cal N}_2$ are two-dimensional invariant
submanifolds of the system \eqref{x1} with the Hamiltonian \eqref{x2}.
\end{theorem}
\begin{proof}
To prove the first statement of the theorem, we need to find points of phase space at which
the rank of the map is not maximal. By direct calculations it can be verified that the Jacobi matrix
of the momentum map has zero minors of the second order at points $\zeta \in {\cal P}$ whose coordinates satisfy
the equations of the system \eqref{y0} and \eqref{y2}, whence ${\cal C} = {\cal N}_1\cup{\cal N}_2$. The invariance of relations ${\cal N}_1$ and ${\cal N}_2$ can be verified by using the following chains of equalities:
\begin{equation*}
\begin{array}{l}
(x_1+x_2)^{\prime}=\{x_1+x_2,H\}=\sigma_1(y_1+y_2)=0;\\
(y_1+y_2)^{\prime}=\{y_1+y_2,H\}=\sigma_2(x_1+x_2)=0;\\
{F_1}^\prime=\{F_1,H\}\bigl|_{F_1=0}=\sigma_3 F_2;\quad {F_2}^\prime=\{F_2,H\}\bigl|_{F_1=0}=\sigma_4 F_2,
\end{array}
\end{equation*}
where $\sigma_k$ are some functions of the phase variables.
\end{proof}

To determine the bifurcation diagram $\Sigma $, it is convenient to go to the polar coordinates
\begin{equation*}
\label{y4}
x_1 = r_1\cos\theta_1,\quad y_1 = r_1\sin\theta_1,\quad
x_2 = r_2\cos\theta_2,\quad y_2 = r_2\sin\theta_2.
\end{equation*}

Substitution in \eqref{y0} and \eqref{y2} leads to the equation $\sin(\theta_1-\theta_2)=0 $, i.e. $\theta_1-\theta_2=0$ and $\theta_1-\theta_2=\pi$. The first possibility, unlike the dynamics of two vortex intensities of opposite signs \cite{SokRyabRCD2017}, is not realized for any positive values of the parameter $c$. For the second case, i.e. when $\theta_1=\theta_2+\pi$, we obtain
\begin{equation}
\label{r1}
\left\{\begin{array}{l}
\theta_1=\theta_2+\pi;\\
\left[\begin{array}{l}
r_1=r_2;\\
(r_1^2+r_2^2)(r_1r_2+c)-(c-2)r_1^2r_2^2-c=0.
\end{array}\right.
\end{array}\right.
\end{equation}

The last equation of the system \eqref{r1} can be rewritten as
\begin{equation}
\label{r2}
-r_1r_2(r_1+r_2)^2+c(1-r_1^2)(1-r_2^2)=0.
\end{equation}
In such a form it coincides with an equation in paper \cite{kevrek2013} on p.~225301-2, derived entirely from other considerations. Our conclusion thus explains that the \cite{kevrek2013} equation on p.~225301-2 determines the radii for critical vortex motions.

The corresponding bifurcation diagram $\Sigma$ is defined on the plane
$\mathbb R^2(f, h)$ and consists of two curves $\gamma_1$ and $\gamma_2$, where
\begin{equation}\label{x2_5}
\begin{array}{l}
\displaystyle{\gamma_1: h=\ln\Bigl(1-\dfrac{f}{2}\Bigr)-\frac{c}{2}\ln(2f),\quad 0<f<2;}\\[3mm]
\gamma_2: \left\{
\begin{array}{l}
\displaystyle{h=\frac{1}{2}\ln\left[\frac{s^2(s-1)}{c+s-1}\right]-\frac{1}{2}c\ln\left[\frac{cs^2}{c+s-1}\right],}\\[3mm]
\displaystyle{f=\frac{cs^2-2(s-1)(c+s-1)}{c+s-1},}
\end{array}\right.\qquad s\in \left(1;\frac{2(1+\sqrt{c})}{2+\sqrt{c}}\right].
\end{array}
\end{equation}
For the values of the phy\-si\-cal pa\-ra\-me\-ter $c>3 $, the curve $\gamma_2$ has a cusp at  $s=\tfrac{\bigl[2-c+\sqrt{c(c-2)}\bigr](c-1)}{c-2}$, which coincides with the point of tangency, when $c=3$ and $s=\tfrac{2(1+\sqrt{c})}{2+\sqrt{c}}$.

As an application, we investigate the  stability of critical trajectories that lie in the preimage of the bifurcation curves \eqref{x2_5}. In this case, it suffices to determine the type (elliptic/hyperbolic) at any one point $(f, h)$ of the smooth branch curve $\Sigma$ \cite{BolBorMam1}.

The type of the critical point $x_0$ of rank $1$ in an integrable system with two degrees of freedom can be calculated with the following
way. You must specify the first integral $F$, such that $dF(x_0)=0$ and
$ dF \ne 0$ in a neighborhood of this point. The $x_0$ is a
fixed point for the Hamiltonian vector field $\sgrad F$ and it is possible to calculate the linearization of this field at a given point - the operator $A_F$ at the point $x_0$. This operator will have two zero eigenvalues, the remaining factor of the characteristic polynomial is $\mu^2-C_F$, where $C_F=\frac {1}{2}\trace(A_F ^ 2)$. When $C_F<0 $ we get a point of the type ``center'' (the corresponding periodic solution is elliptic, it is a stable periodic solution in phase space, the limit of the concentric family of two-dimensional regular tori), and for $C_F>0$ we get a point of the type ``saddle'' (the corresponding periodic solution has a hyperbolic type, there are motions asymptotic to this solution lying on two-dimensional separatrix surfaces). Here we present explicit expressions for $C_F$ only for bifurcation curves $\gamma_1$ and $\gamma_2$:
\begin{equation*}\label{x2_6}
\begin{array}{l}
\gamma_1: C_F=(4-c)f^2+4cf-4c,\quad 0<f<2;\\
\gamma_2: C_F=(c-2)s^2+2(c-1)(c-2)s-2(c-1)^2, \quad s\in \left(1;\frac{2(1+\sqrt{c})}{2+\sqrt{c}}\right].
\end{array}
\end{equation*}

Fig.~\ref{fig1} shows an enlarged fragment of the bifurcation diagram in the case of the the identical
intensities  while the parameter $c>3$. The signs $``+``$ and $``-``$ correspond to elliptic (stable) and hyperbolic (unstable) periodic solutions in the phase space. As expected, the type change occurs at the cusp $A$ and at the point of tangency $B$ of the bifurcation diagram $\Sigma$.

\begin{figure}[!ht]
\centering
\includegraphics[width=0.7\textwidth]{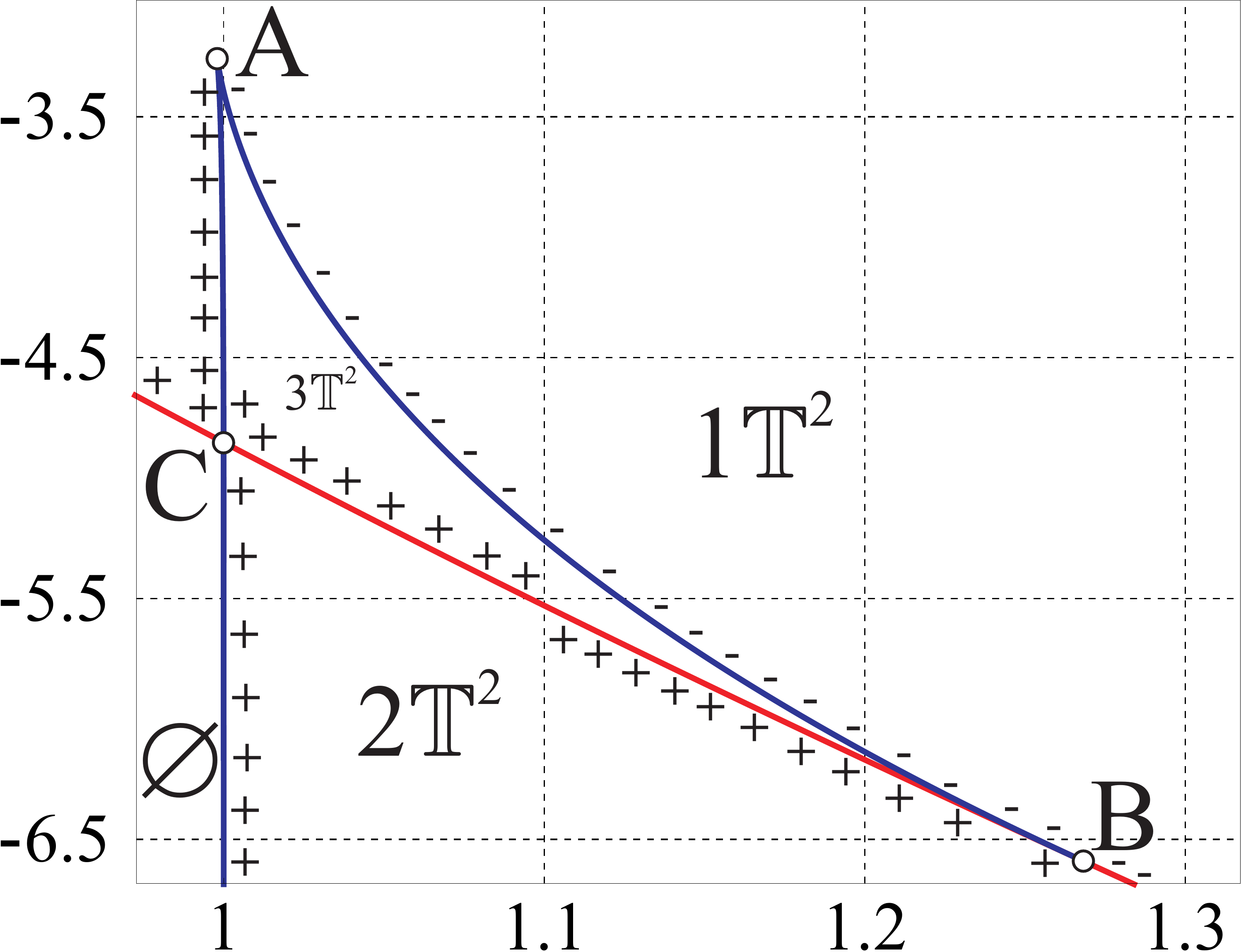}
\caption{Enlarged fragment of the bifurcation diagram $\Sigma$.}
\label{fig1}
\end{figure}

\section{Bifurcation $3\mathbb T^2 \to \mathbb S^1\times\left(\mathbb S^1\,\dot{\cup}\,\mathbb S^1\,\dot{\cup}\,\mathbb S^1\right)\to \mathbb T^2$}
Here we confine to identical intensities ($\Gamma_1=\Gamma_2=1$). Perform an explicit reduction to a system with one degree of freedom. For this, in the system \eqref{x1} ($N=2$) with the Hamiltonian \eqref{x2}, we change phase variables $(x_k, y_k)$ by the new variables $(u,v,\alpha)$ using the formulas:
\begin{equation*}\label{z1}
\begin{array}{l}
x_1=u\cos(\alpha)-v\sin(\alpha),\quad y_1=u\sin(\alpha)+v\cos(\alpha),\\[3mm]
x_2=\sqrt{f-u^2-v^2}\cos(\alpha),\quad y_2=\sqrt{f-u^2-v^2}\sin(\alpha).
\end{array}
\end{equation*}

The physical variables $(u,v)$ are the Cartesian coordinates of one of the vortices in the coordinate system associated with another vortex rotating around the center of vorticity. The choice of such variables is suggested by the presence of the integral of the moment of vorticity \eqref{x5}, which is invariant under the rotation group $SO(2)$. The existence of a one-parameter symmetry group allows one to reduce to a system with one degree of freedom, just as it does in mechanical systems with symmetry \cite{Kharlamov1988}. Backward replacement
\begin{equation*}
U=\frac{x_1x_2+y_1y_2}{\sqrt{x_2^2+y_2^2}},\quad V=\frac{y_1x_2-x_1y_2}{\sqrt{x_2^2+y_2^2}}
\end{equation*}
leads to canonical variables with respect to the bracket \eqref{x4} ($\Gamma_1=\Gamma_2=1$):
\begin{equation*}
\{U,V\}=-\{V,U\}=1,\quad \{U,U\}=\{V,V\}=0.
\end{equation*}

The system with respect to the new variables $(u,v)$ is Hamiltonian
\begin{equation}\label{z2}
u^\prime=\frac{\partial H_1}{\partial v},\quad
v^\prime=-\frac{\partial H_1}{\partial u}
\end{equation}
with Hamiltonian
\begin{equation}
\label{z3}
\begin{array}{l}
\displaystyle{H_1=\frac{1}{2}\ln(1-u^2-v^2)+\frac{1}{2}\ln(1-f+u^2+v^2)-\frac{1}{2}c\ln[f-2\sqrt{f-u^2-v^2}\,u].}
\end{array}
\end{equation}

The rotation angle $\alpha(t)$ of the rotating coordinate system satisfies the differential equation
\begin{equation*}
\displaystyle{\alpha^\prime=\frac{R(u,v)}{Q(u,v)},}
\end{equation*}
where
\begin{equation*}
\begin{array}{l}
R(u,v)=2\bigl\{\sqrt{f-u^2-v^2}[b(1-f+u^2+v^2)+f]-u[(b-2)(u^2+v^2-f)+b]\bigr\},\\[3mm]
Q(u,v)=\sqrt{f-u^2-v^2}(1-f+u^2+v^2)[f-2\sqrt{f-u^2-v^2}u].
\end{array}
\end{equation*}

The fixed points of the reduced system \eqref{z2} are determined by the critical points of the reduced Hamiltonian \eqref{z3} and correspond to the relative equilibria of vortices in the system \eqref{x1}. For a fixed value of the integral of the moment of vorticity $f$, the regular levels of the reduced Hamiltonian are compact and motions occur along closed curves.
It can be shown that the critical values of the reduced Hamiltonian determine the bifurcation diagram \eqref {x2_5}. For a segment of the bifurcation curve $(AB)$ (Fig.~\ref{fig1}), the motion on the plane $(u,v)$ occurs along a curve that is topologically structured as $\mathbb S^1\,\dot{\cup}\,\mathbb S^1\,\dot{\cup}\,\mathbb S^1$ (Fig.~2b)), and the integral critical surface is a trivial bundle over $\mathbb S^1$ with a layer $\mathbb S^1\,\dot{\cup}\,\mathbb S^1\,\dot{\cup}\,\mathbb S^1$.
\begin{figure}[!ht]
\centering
\includegraphics[width=1\textwidth]{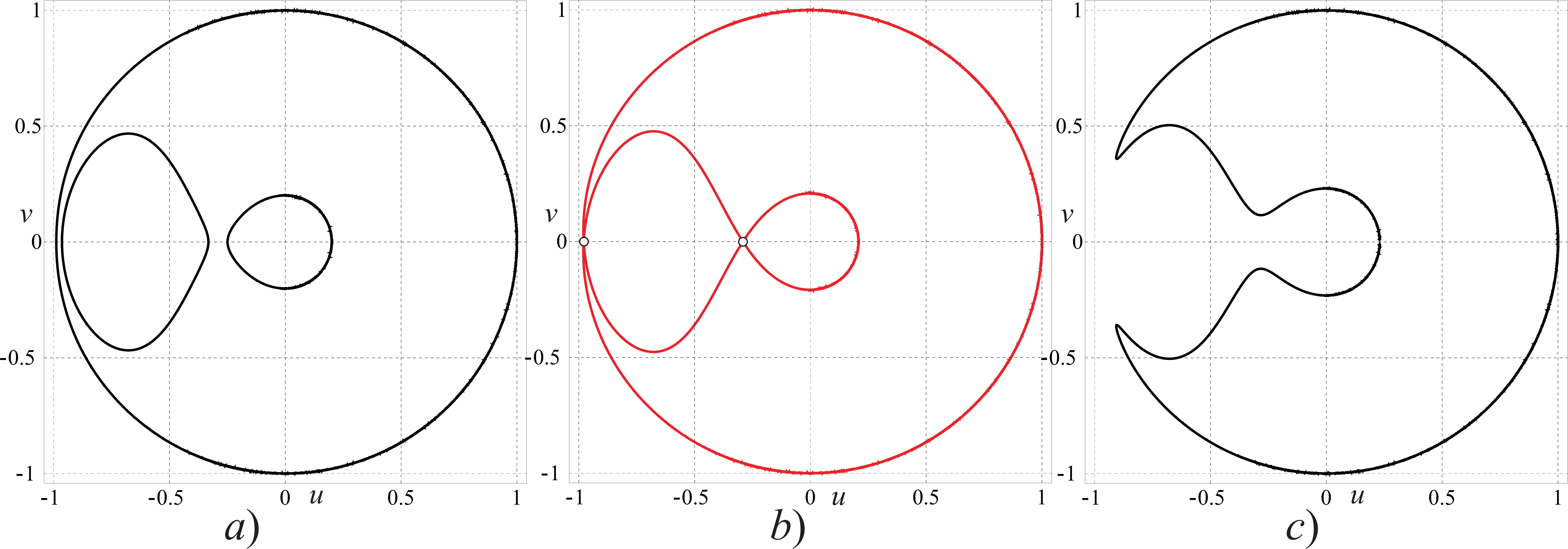}
\caption{Level curves of reduced Hamiltonian $H_1$ for $c>3$.}
\label{fig2}
\end{figure}

When passing through the section of the curve $(AB)$ of the bifurcation diagram $\Sigma$ for $c>3$ (Fig.~\ref{fig1}), the bifurcation of three tori into one $3\mathbb T^2 \to \mathbb S^1\times\left(\mathbb S^1\,\dot{\cup}\,\mathbb S^1\,\dot{\cup}\,\mathbb S^1\right)\to \mathbb T^2$. With the help of the level curves of the reduced Hamiltonian in Fig.~\ref{fig2}, this bifurcation is clearly demonstrated ($h_1=-4.5$ for a) $f_1=1.04$; b) $f_{2}=1.042957$; c) $f_3=1.05$).

After splitting the phase space into areas in which the number of tori remains unchanged, explicitly defining the bifurcation diagram and the bifur\-cations of Liouville tori, we can formulate the problem of classifying the absolute motions of vortices, as well as determining the topological type of three-dimensional isoenergy manifolds.

\bigskip
The authors thanks for valuable discussions to A.V. Borisov.

\end{document}